\documentclass[conference,letterpaper]{IEEEtran}

\usepackage[utf8]{inputenc} 
\usepackage[T1]{fontenc}
\usepackage{url}
\usepackage{ifthen}
\usepackage{cite}
\usepackage[cmex10]{amsmath} 

\usepackage{enumerate}
\usepackage{amsmath}
\usepackage{amssymb,latexsym}
\usepackage{amsthm}
\usepackage{color}
\usepackage{cancel}
\usepackage{graphicx}
\usepackage{cite}
\usepackage{changes}
\usepackage{hyperref}
\usepackage{comment}
\usepackage{amscd}
\usepackage{mathtools}

\newtheorem{theorem}{Theorem}[section]

\newtheorem{lemma}[theorem]{Lemma}
\newtheorem{corollary}[theorem]{Corollary}
\newtheorem{definition}[theorem]{Definition}
\newtheorem{proposition}[theorem]{Proposition}

\newtheorem{remark}[theorem]{Remark}

\newtheorem{construction}[theorem]{Construction}

\DeclareMathOperator{\C}{\mathcal{C}}

\newcommand{\fq}{\mathbb{F}_{q}}

\newcommand{\F}{{\mathbb F}}

\newcommand{\la}{\langle}
\newcommand{\ra}{\rangle}

\newcommand{\PG}{\mathrm{PG}}
\newcommand{\GL}{\mathrm{GL}}

\newcommand{\fqm}{\mathbb{F}_{q^m}}

\begin{document}
\title{Two-weight rank-metric codes} 


\author{\IEEEauthorblockN{Ferdinando Zullo\IEEEauthorrefmark{1}, Olga Polverino\IEEEauthorrefmark{1}, Paolo Santonastaso\IEEEauthorrefmark{1} and John Sheekey\IEEEauthorrefmark{2}}\\
\IEEEauthorblockA{\IEEEauthorrefmark{1}Dipartimento di Matematica e Fisica, Universit\`{a} degli \\Studi della Campania ``Luigi Vanvitelli", Caserta, Italy}\\
\IEEEauthorblockA{\IEEEauthorrefmark{2}School of Mathematics and Statistics,\\ 
University College Dublin, Dublin, Ireland }
Corresponding author: F. Zullo (email: ferdinando.zullo@unicampania.it) }


\maketitle

\begin{abstract}
Two-weight linear codes are linear codes in which any nonzero codeword can have only two possible distinct weights. Those in the Hamming metric have proven to be very interesting for their connections with authentication codes, association schemes, strongly regular graphs, and secret sharing schemes. In this paper, we characterize two-weight codes in the rank metric, answering a recent question posed by Pratihar and Randrianarisoa.
\end{abstract}

\section{Introduction}

Rank-metric codes have garnered significant attention in recent decades, owing to their  applications and intriguing connections with various mathematical objects, such as in linear random network coding, see \cite{bartz2022rank}. Nevertheless, the origin of rank-metric codes traces its roots to Delsarte \cite{de78} in  1978, and later to Gabidulin in \cite{ga85a} and Roth in \cite{roth1991maximum}.

Most of the focus on the study of rank-metric codes has regarded the linear case, that is $\F_{q^m}$-subspaces of $\F_{q^m}^n$, where $\F_{q^m}^n$ is equipped with the rank distance, defined as follows: for any $x=(x_1,\ldots,x_n),y=(y_1,\ldots,y_n) \in \F_{q^m}^n$ then
\[d(x,y)=w(x-y)=\dim_{\fq}(\langle x_1-y_1,\ldots,x_n-y_n \rangle_{\fq}),\]
where $w(x-y)$ denotes the \textbf{rank weight} of $x-y$.

In this paper, a \textbf{rank-metric code} is any $\F_{q^m}$-subspace of $\F_{q^m}^n$.
Let $\C \subseteq \F_{q^m}^n$ be a rank-metric code. We will write that $\C$ is an $[n,k,d]_{q^m/q}$ code if $k$ is the $\F_{q^m}$-dimension of $\C$ and $d$ is its minimum distance, that is 
\[
d=\min\{d(x,y) \colon x, y \in \C, x \neq y  \}.
\]
When the minimum distance is not known or is irrelevant, we simply write $[n, k]_{q^m/q}$.

It is possible to prove a Singleton-like bound for a rank-metric code.

\begin{theorem}(see \cite{de78}) \label{th:singletonrank}
    Let $\C \subseteq \F_{q^m}^n$ be an $[n,k,d]_{q^m/q}$ code.
Then 
\begin{equation}\label{eq:boundgen}
mk \leq \max\{m,n\}(\min\{n,m\}-d+1).\end{equation}
\end{theorem}

An $[n,k,d]_{q^m/q}$ code is called \textbf{Maximum Rank Distance code} (or shortly \textbf{MRD code}) if its parameters attains the bound \eqref{eq:boundgen}.
A family of MRD codes has been shown in \cite{de78,ga85a,roth1991maximum}.

In the classical Hamming-metric setting, well-studied codes are those such that the weight of its nonzero elements can only assume two distinct values; see the well celebrated paper by Calderbank and Kantor \cite{calderbank1986geometry}. These codes present a wide variety of properties, and their geometry can be very different, although some classification results have been proved; see e.g. \cite{jungnickel2018classification}.

In this paper we analyze rank-metric codes $\C$ which are \textbf{two-weight} codes, that is the rank of its nonzero elements can only assume two distinct values.
An example of two-weight code in the rank metric is given by MRD codes in $\F_{q^m}^n$ with $n\leq m$ and with minimum distance $n-1$.
The study of two-weight rank-metric codes was initiated by Pratihar and Randrianarisoa in \cite{pratihar2023antipodal}, considering \emph{antipodal} two-weight rank-metric codes $\C$. An $[n,k]_{q^m/q}$ code is an \textbf{antipodal} code if there exists a codeword of weight $n$. They ask whether or not all of two-weight antipodal codes \emph{arise} from MRD codes. Using a geometric description of rank-metric codes developed in \cite{Randrianarisoa2020ageometric} and also in \cite{sheekeygeo} and by a clever characterization in terms of spreads, they prove a complete classification of antipodal two-weight rank-metric codes when the minimum rank distance is equal to half of the length. 
More precisely, they proved that such a code is always induced by an MRD code over an extension of $\fq$ (see \cite{pratihar2023antipodal} and \cite{polverino2023divisible}).
This proved that in this case, two-weight codes in the rank metric have a strong geometric structure.

In this paper we give a more general geometric characterisation of antipodal two-weight codes, proving that, when $q$ is large enough, all the two-weight rank-metric codes arise from well-studied combinatorial objects known as \emph{scattered spaces}. Indeed, in Section \ref{sec:5}, we prove that this is the case for all (not necessarily antipodal) two-weight codes. Moreover, using the lower bound on the dimension of scattered subspaces and their existence results, we are able to give a lower bound on the length of a two-weight code and a range of values for its length in which we can guarantee the existence of such codes.

\section{Preliminaries}

In this section we briefly recall some definitions and results we will need in the paper.

\subsection{Linear sets}\label{sec:linset}

Let $V$ be a $k$-dimensional vector space over $\F_{q^m}$ and let $\Lambda=\PG(V,\F_{q^m})=\PG(k-1,q^m)$.
Recall that, if $U$ is an $\fq$-subspace of $V$ of dimension $n$, then the set of points
\[ L_U=\{\la { u} \ra_{\mathbb{F}_{q^m}} : { u}\in U\setminus \{{ 0} \}\}\subseteq \Lambda \]
is said to be an $\fq$-\textbf{linear set of rank $n$}.
Let $P=\langle v\rangle_{\F_{q^m}}$ be a point in $\Lambda$. The \textbf{weight of $P$ in $L_U$} is defined as 
\[ w_{U}(P)=\dim_{\fq}(U\cap \langle v\rangle_{\F_{q^m}}). \]
If $N_i$ denotes the number of points of $\Lambda$ having weight $i\in \{0,\ldots,n\}$  in $L_U$, the following relations hold:
\begin{equation}\label{eq:card}
    |L_U| \leq \frac{q^n-1}{q-1},
\end{equation}
\begin{equation}\label{eq:pesicard}
    |L_U| =N_1+\ldots+N_n,
\end{equation}
\begin{equation}\label{eq:pesivett}
    N_1+N_2(q+1)+\ldots+N_n(q^{n-1}+\ldots+q+1)=q^{n-1}+\ldots+q+1.
\end{equation}

Furthermore, $L_U$ is called \textbf{scattered} if it has the maximum number $\frac{q^n-1}{q-1}$ of points, or equivalently, if all points of $L_U$ have weight one. Moreover, for a positive integer $e \mid m$, with $e<m$, we say that an $\F_{q^e}$-subspace $U$ of $V$ is an \textbf{$\F_{q^e}$-scattered subspace} of $V$ if  
\[
\dim_{\F_{q^e}}(U \cap \langle v \rangle_{\F_{q^m}}) \leq 1,
\]
for any $v \in V$.
Blokhuis and Lavrauw provided the following bound on the rank of a scattered liner set.

\begin{theorem}(see \cite{blokhuis2000scattered})\label{th:boundscattrank}
Let $L_U$ be a scattered $\fq$-linear set of rank $n$ in $\PG(k-1,q^m)$, then
\[ n\leq \frac{mk}2. \]
\end{theorem}

A scattered $\fq$-linear set of rank $km/2$ in $\PG(k-1,q^m)$ is said to be a \textbf{maximum scattered} and $U$ is said to be a maximum scattered $\fq$-subspace as well. 

In the next result we summarize what is known on the existence of maximum scattered linear sets/subspaces.

\begin{theorem}(see \cite{blokhuis2000scattered,ball2000linear,bartoli2018maximum,csajbok2017maximum})\label{th:existencemaxscatt}
If $mk$ is even then there exist maximum scattered subspaces in $\F_{q^m}^k$.
\end{theorem}

If $L_U$ is a scattered $\fq$-linear set for which $L_U$ is not contained in another scattered $\fq$-linear set is said to be \textbf{maximally scattered} and $U$ is called a \textbf{maximally scattered} subspace. Clearly, maximum scattered linear sets are maximally scattered, but the converse does not hold in general; see e.g. \cite{lavrauw2013scattered}.

Now, we recall the notion of the dual of a linear set.
Let $\sigma \colon V \times V \rightarrow \F_{q^m}$ be a nondegenerate reflexive bilinear form on the $k$-dimensional $\F_{q^m}$-vector space $V$ and consider \[
\begin{array}{cccc}
    \sigma': & V \times V & \longrightarrow & \F_q  \\
     & (x,y) & \longmapsto & \mathrm{Tr}_{q^m/q} (\sigma(x,y)).
\end{array}
\] 
So, $\sigma'$ is a nondegenerate reflexive bilinear form on $V$ seen as an $\fq$-vector space of dimension $km$. Then we may consider $\perp$ and $\perp'$ as the orthogonal complement maps defined by $\sigma$ and $\sigma'$. For an $\F_q$-linear set in $\PG(V,\F_{q^m})$ of rank $n$, the $\F_q$-linear set $L_{U^{\perp'}}$ in $\PG(V,\F_{q^m})$ of rank $km-n$ is, called the \textbf{dual linear set of $L_U$} with respect to $\sigma$.
The definition of the dual linear set does not depend on the choice of $\sigma$; see e.g. \cite[Proposition 2.5]{polverino2010linear}.

Moreover, we have the following relation between the weight of a subspace with respect to a linear set and the weight of its polar space with respect to the dual linear set.

\begin{proposition} (see \cite[Property 2.6]{polverino2010linear}) \label{prop:weightdual}
Let $L_U \subseteq \PG(k-1,q^m)$ be a linear set of rank $n$. Let $\Omega_s=\PG(W,\F_{q^m})$ be an $s$-dimensional projective space of $\PG(k-1,q^m)$.
Then 
\[
w_{U^{\perp'}}(\Omega_s^{\tau})=w_{U}(\Omega_s)+km-n-(s+1)m,
\]
i.e.
\[ \dim_{\fq}(U^{\perp'}\cap W^{\perp})=\dim_{\fq}(U\cap W)+\dim_{\fq}(V)\]
\[-\dim_{\fq}(U)-\dim_{\fq}(W). \]
\end{proposition}

We refer to \cite{lavrauw2015field} and \cite{polverino2010linear} for comprehensive references on linear sets.

\subsection{Rank-metric codes and $q$-systems}\label{sec:qsystems}

Now, we recall the definition of equivalence between rank-metric codes in $\F_{q^m}^n$. An $\F_{q^m}$-linear isometry $\phi$ of $\F_{q^m}^n$ is an $\F_{q^m}$-linear map of $\F_{q^m}^{n}$ that preserves the rank weight, i.e. $w(x)=w(\phi(x))$, for every $x \in  \F_{q^m}^{n}$, or equivalently $d(x,y)=d(\phi(x),\phi(y))$, for every $x,y \in  \F_{q^m}^{n}$.
It is known that the group of $\F_{q^m}$-linear isometries of $\F_{q^m}^n$ equipped with rank distance is generated by the (nonzero) scalar
multiplications of $\F_{q^m}$ and the general linear group $\mathrm{GL}(n,\F_q)$, see e.g. \cite{berger2003isometries}. So we say that two rank-metric codes $\C,\C' \subseteq \F_{q^m}^n$ are (linearly) equivalent if there exists an isometry $\phi$ such that $\phi(\C)=\C'$.
Clearly, when studying equivalence of $[n, k]_{q^m/q}$ codes the
action of $\F_{q^m}^*$ is trivial. This means that two $[n,k]_{q^m/q}$ codes $\C$ and $\C'$ are equivalent if and only if
there exists $A \in \mathrm{GL}(n,q)$ such that
$\C'=\C A=\{vA : v \in \C\}$. 
We will consider codes that are \emph{nondegenerate}.

 \begin{definition}
An $[n,k]_{q^m/q}$ code $\C$ is said to be \textbf{nondegenerate} if the columns of any generator matrix of $\C$ are $\fq$-linearly independent. 
\end{definition}

The geometric counterpart of rank-metric are the systems. 
 \begin{definition}
An $[n,k,d]_{q^m/q}$ \textbf{system} $U$ is an $\F_q$-subspace of $\F_{q^m}^k$ of dimension $n$, such that
$ \langle U \rangle_{\F_{q^m}}=\F_{q^m}^k$ and
$$ d=n-\max\left\{\dim_{\F_q}(U\cap H) \mid H \textnormal{ is an $\F_{q^m}$-hyperpl. of }\F_{q^m}^k\right\}.$$
Moreover, two $[n,k,d]_{q^m/q}$ systems $U$ and $U'$ are \textbf{equivalent} if there exists an $\F_{q^m}$-isomorphism $\varphi\in\GL(k,\F_{q^m})$ such that
$$ \varphi(U) = U'.$$
We denote the set of equivalence classes of $[n,k,d]_{q^m/q}$ systems by $\mathfrak{U}[n,k,d]_{q^m/q}$.
\end{definition}

\begin{theorem}(see \cite{Randrianarisoa2020ageometric}) \label{th:connection}
Let $\C$ be an $[n,k,d]_{q^m/q}$ rank-metric code and let $G$ be an its generator matrix.
Let $U \subseteq \F_{q^m}^k$ be the $\F_q$-span of the columns of $G$.
The rank weight of an element $x G \in \C$, with $x \in \F_{q^m}^k$ is
\begin{equation}\label{eq:relweight}
w(x G) = n - \dim_{\fq}(U \cap x^{\perp}),\end{equation}
where $x^{\perp}=\{y \in \F_{q^m}^k \colon x \cdot y=0\}$ and $\cdot$ denotes the standard dot product in $\F_{q^m}^k$. In particular,
\begin{equation} \label{eq:distancedesign}
d=n - \max\left\{ \dim_{\fq}(U \cap H)  \colon H\mbox{ is an } \F_{q^m}\mbox{-hyperpl. of }\F_{q^m}^k  \right\}.
\end{equation}
\end{theorem}

Thanks to Theorem \ref{th:connection}, it can be proved that there is a one-to-one  correspondence  between  equivalence  classes of $[n,k,d]_{q^m/q}$ codes and equivalence classes of $[n,k,d]_{q^m/q}$ systems; see \cite{alfarano2021linear,Randrianarisoa2020ageometric}.

Finally, we recall that by using the dual of a system associated to a code, it is possible to define a new code as follows, see \cite{borello2023geometric}.

\begin{definition}
Let $\C$ be an $[n,k,d]_{q^m/q}$ and let $U$ be a system associated with $\C$. Assume that $\dim_{\F_q}(U\cap \langle v\rangle_{\F_{q^m}})<m$, for any $v \in \F_{q^m}^k$.
Then a \textbf{geometric dual} of $\C$ (with respect to $\perp'$) is defined as $\C'$, where $\C'$ is any code associated with the system $U^{\perp'}$.
\end{definition}

Note that the geometric dual of an $[n,k]_{q^m/q}$ code is an $[mk-n,k]_{q^m/q}$ code.

\section{Auxiliary results on linear sets}

The problem of determining bounds on the size of a linear set is a hard problem. 
The following is a lower bound on the number of points of a linear set admitting a point of weight one (which generalizes \cite[Lemma 2.2]{bonoli2005fq}). 

\begin{theorem} (see \cite[Theorem 1.2]{de2019minimum}) \label{th:minsize}
Let $L_U$ be an $\F_q$-linear set in $\PG(1,q^m)$ of rank $n$, such that there exists $P \in L_U$ with $w_{U}(P)=1$. 
Then
\[ |L_U|\geq q^{n-1}+1. \]
\end{theorem}

We also refer to \cite{adriaensen2023minimum} for more bounds in higher dimensional projective spaces. We also recall the following recent result, which regards the linearity of a linear set.

\begin{theorem} {(see \cite[Theorem 2]{csajbok2023maximum})} \label{th:geometricfieldoflinearity}
Let $L_U$ be an $\F_q$-linear set of rank $n$ in $\PG(k-1,q^m)$ with $n\leq (k-1)m$ and $m\leq q$.
Let $e=\min\{w_{U}(P) \colon P \in L_U \}$ and suppose $e>1$. Then there exists a positive integer $t \geq e$ such that $t \lvert m$ and $L_U=L_W$ where $W=\langle U \rangle_{\F_{q^t}}$.
\end{theorem}

In the following we use the above mentioned results to prove that if an $\fq$-subspace $U$ of $\F_{q^m}^2$ meets every one-dimensional $\F_{q^m}$-subspace of $\F_{q^m}^2$ either trivially or in an $\fq$-subspace of dimension $e>0$, then $e \mid m$ and $U$ is $\F_{q^e}$-linear, as well. 

\begin{corollary}\label{cor:linearityoneweight}
Assume that $q \geq m$. Let $U$ be an $n$-dimensional $\fq$-subspace of a $k$-dimensional vector space $V$ over $\F_{q^m}$. Suppose that $n \leq (k-1)m$, $\langle U \rangle_{\F_{q^m}}=V$ and
\[ \dim_{\fq}(U \cap \langle w \rangle_{\F_{q^m}})\in \{0,e\}, \]
for any $w \in V$, then $e \mid n$, $e \mid m$ and $U$ is an $\F_{q^e}$-subspace of dimension $n/e$. In other words, $U$ is an $\F_{q^e}$-scattered subspace of $V$.
\end{corollary}
\begin{proof}
The assertion is clear if $e=1$. So suppose that $1<e<m$. First, assume that $k=2$. Let consider the $\fq$-linear set $L_U$ which, by the assumption on $U$, has size $\frac{q^n-1}{q^e-1}$ and so $e \mid n$ with $e<n$.
By Theorem \ref{th:geometricfieldoflinearity}, it follows that there exists a positive integer $t \geq e$ such that $t \lvert m$ and such that $L_U=L_W$, where $W=\langle U \rangle_{\F_{q^t}}$. Denote by $n'=\dim_{\F_{q}}(W)=t g$ and note that $n'\geq n$ and $t<n'$ (i.e. $g>1$).
We may also assume that there exists a point $P=\langle u \rangle_{\F_{q^m}} \in \PG(1,q^m)$ such that 
\[ w_{W}(P)=\dim_{\F_{q^t}}(W \cap \langle u \rangle_{\F_{q^m}})=1, \]
where $L_W$ is seen as an $\F_{q^t}$-linear set, otherwise we may apply Theorem \ref{th:geometricfieldoflinearity} again to $L_W$. In particular, applying Theorem \ref{th:minsize} and the bound \eqref{eq:card} we obtain
\begin{equation}\label{eq:doublebound} (q^{t})^{g-1}+1\leq \frac{q^n-1}{q^e-1}\leq \frac{q^{gt}-1}{q^t-1}.  \end{equation}
From the former bound in \eqref{eq:doublebound} we obtain
\[ q^{t(g-1)+e}+q^e\leq q^{n}+q^{t(g-1)}, \]
which implies that $n\geq t(g-1)+e$, that is
\begin{equation}\label{eq:bound1weightpoints}
n-e \geq n'-t,
\end{equation}
and also 
\begin{equation} \label{eq:bound1weightpoints1}
    n+t \geq n'+e \geq n'+2.
\end{equation}
By the latter bound in \eqref{eq:doublebound} we obtain
\begin{equation}\label{eq:bound2weightpoints}
q^{n+t}+q^e+q^{n'}\leq q^{n'+e}+q^t+q^n.
\end{equation}
Taking into account \eqref{eq:bound1weightpoints1} and that $n'+e \geq \max\{n'+2,t\}$, it follows that $n'+e\geq n+t$, that, together with \eqref{eq:bound1weightpoints}, implies that
\[ n-e=n'-t. \]
As a consequence, since by \eqref{eq:bound1weightpoints1} $n \geq t(g-1)+2 \geq t$, \eqref{eq:bound2weightpoints} reads as
\[ q^e+q^{n'}\leq q^t+q^n, \]
implies that $n'\leq n$ and hence the equality and $t=e$.
Since $U$ is contained in $W$ and have the same dimension, we have $U=W$ and hence $U$ is an $\F_{q^e}$-subspace. So the assertion is proved for $k=2$.  Now, assume that $k>2$. For any two vectors $u,v \in U$ which are not proportional over $\F_{q^m}$ and consider 
\[ U'_{u,v}=\langle u,v \rangle_{\F_{q^m}}\cap U. \]
Then $U'_{u,v}$ is an $\fq$-subspace with the property that 
\[ \dim_{\F_q}(U'_{u,v} \cap \langle w \rangle_{\F_{q^m}})\in \{0,e\}, \]
for any nonzero $w \in \langle u,v \rangle_{\F_{q^m}}$.
Applying the case $k=2$ to $U'_{u,v}$, from which we obtain that $e \mid m$ and $U'_{u,v}$ is an $\F_{q^{e}}$-subspace of $V$ for every $\F_{q^m}$-linearly independent $u,v \in U$ and hence $U$ is an $\F_{q^{e}}$-subspace as well.
\end{proof}

The above result cannot be extended to $\F_q$-subspaces meeting every one-dimensional $\F_{q^m}$-subspaces in $\fq$-subspaces of dimension a multiple of a fixed integer $e$.

\begin{remark}\label{Rk:ex2}
Let $\lambda \in \F_{q^4}$ such that $\F_q(\lambda)=\F_{q^4}$. Let $S_1=\langle 1,\lambda \rangle_{\F_q}$. Let $S_2$ be a $3$-dimensional $\F_{q^4}$-subspace of $\F_{q^{16}}$ and consider 
\[
U= S_1 \times S_2 =\{(s_1,s_2) : s_1 \in S_1, s_2 \in S_2\} \subseteq \F_{q^{16}}^2.
\]
Then $U$ is a $14$-dimensional $\F_q$-subspace of $\F_{q^{16}}^2$ and note that $U$ cannot be also an $\F_{q^i}$-subspace for some $i>1$. 
Moreover,
\[ \dim_{\fq}(U \cap \langle w \rangle_{\F_{q^{16}}})\in \{0,2,12\}, \]
for every $w \in \F_{q^{16}}^2$. 
Therefore, $U$ is an $\fq$-subspace meeting every one-dimensional $\F_{q^{16}}$-subspace of $\F_{q^{16}}^2$ in a subspace of dimension a multiple of $2$ but $U$ is strictly $\fq$-linear. 
\end{remark}

\section{Two-weight linear rank-metric codes}\label{sec:5}

In this section we completely describe two-weight codes. A way to construct two-weight codes is via scattered subspaces.

\begin{construction}
Let $U$ be an $\F_{q^e}$-scattered subspace in $\F_{q^m}^k$ such that $\dim_{\F_{q}}(U)=n$ and $\langle U \rangle_{\F_{q^m}}=\F_{q^m}^k$. A code $\C$ having $U$ as an its associated system is said to be an \textbf{$\F_{q^e}$-scattered code}. 
\end{construction}

Note that, when $k=2$, an $\F_{q^e}$-scattered code $\C$ is an $[n,2,n-e]_{q^m/q}$ code with nonzero weights $n-e$ and $n$. In addition, when $e=1$, $\C$ is MRD for any value of $n\leq m$.
The above construction can be extended via duality as follows.

\begin{construction}\label{constr:viascatt}
Let $U$ be an $\F_{q^e}$-scattered subspace in $\F_{q^m}^k$ such that $\dim_{\F_{q}}(U)=n$ and $\langle U \rangle_{\F_{q^m}}=\F_{q^m}^k$. Consider $U^{\perp'}$ the dual of $U$.
Then $\dim_{\F_q}(U^{\perp'})=km-n$, $\langle U^{\perp'} \rangle_{\F_{q^m}}=\F_{q^m}^k$ and by Proposition \ref{prop:weightdual}
\[ \dim_{\fq} (U^{\perp'}\cap H) \in \{ (k-1)m-n, (k-1)m-n+e \}. \]
Consider $\C$ the rank-metric code associated with $U^{\perp'}$, then $\C$ is a nondegenerate $[km-n,k,m-e]_{q^m/q}$ code with nonzero weights $m-e$ and $m$. In other words, the geometric dual of an $\F_{q^e}$-scattered $[n,k]_{q^m/q}$ code is an $[km-n,k,m-e]_{q^m/q}$ with nonzero weights $m-e$ and $m$.
\end{construction}

\begin{remark}
For an $\F_{q^e}$-scattered subspace $U$ in $\F_{q^m}^k$ such that $\dim_{\F_{q}}(U)=n$ and $\langle U \rangle_{\F_{q^m}}=\F_{q^m}^k$, it results that $\langle U^{\perp'} \rangle_{\F_{q^m}}=\F_{q^m}^k$. Indeed, suppose that $U^{\perp'}$ is contained in a hyperplane of $\F_{q^m}^k$, i.e. $\langle U^{\perp'} \rangle_{\F_{q^m}} \subseteq H$. This means that $H^\perp \subseteq U$ and so $ \dim_{\F_q}(H^\perp \cap U)=m>e$, a contradiction.
\end{remark}

In Construction \ref{constr:viascatt}, it can happen that the obtained codes are antipodal, for instance when $m=n$ and $k=2$.

\begin{remark}\label{rk:MRD}
The $[km-n,k,m-e]_{q^m/q}$ code $\C$ from Construction \ref{constr:viascatt} is an MRD code if and only if
\[ | \C|= q^{\max\{m,km-n\}(\min\{m,km-n\}-d+1)},\]
which can be read as
\begin{equation}\label{eq:2deg} km=(km-n)(e+1), \end{equation}
since $n\leq km/2$ via the bound of Theorem \ref{th:boundscattrank}. Again, by \ref{th:boundscattrank}, \eqref{eq:2deg} implies that 
\[
km=(km-n)(e+1) \geq \frac{km}{2}(e+1),
\]
from which $e=1$. Hence, from \eqref{eq:2deg} it follows that $n=\frac{km}2$.
\end{remark}

We will start by characterizing the antipodal two weight codes. 

Our approach now is to use the geometric correspondence between rank-metric codes and systems of $\F_{q^m}^2$. We will start by recalling some results for the antipodal case from \cite{pratihar2023antipodal}.

\begin{proposition} \label{prop:limitationnmantipodal} Let $\C$ be a nondegenerate antipodal two-weight $[n,k,d]_{q^m/q}$ code. Then $n\leq m$.
\end{proposition}

\begin{theorem}(see \cite[Theorem 2]{pratihar2023antipodal})\label{th:k=2}
Let $\C$ be a nondegenerate $[n,k,d]_{q^m/q}$ code. If $\C$ is an antipodal two-weight code then $k=2$.
\end{theorem}

\begin{lemma}(see \cite[Corollary 1]{pratihar2023antipodal})\label{lem:n-ddividesn}
Let $\C$ be a nondegenerate $[n,2,d]_{q^m/q}$ code. If $\C$ is an antipodal two-weight code then $n-d\mid n$.
\end{lemma}

Now, we are able to characterize antipodal codes, extending \cite{pratihar2023antipodal}.

\begin{theorem}
Assume that $q \geq m.$ Let $\C$ be a nondegenerate antipodal two-weight $[n,2,d]_{q^m/q}$ code with $n=\ell(n-d)$. Then $n-d \mid m$ and $\C$ is an $\F_{q^{n-d}}$-scattered code.
\end{theorem}
\begin{proof}
Let $U\subseteq \mathbb{F}_{q^m}^2$ be a system associated with the code $\C$ and 
\[ \dim_{\fq}(U \cap \langle v \rangle_{\fqm}^{\perp}) \in \{ 0, n-d\}, \]
for any $v \in \fqm^2$.
Note that $\dim_{\fq}(U)=\ell (n-d)$.
By Corollary \ref{cor:linearityoneweight}, it follows that $n-d$ divides $m$ and $U$ is an $\F_{q^{n-d}}$-subspace of $\F_{q^m}^2$ with the property that 
\[ \dim_{\F_{q^{n-d}}}(U\cap \langle  w \rangle_{\F_{q^m}})\leq 1, \]
for any $w \in \F_{q^m}^2$ and hence by Theorem \ref{th:connection} the assertion follows.
\end{proof}

This answers to the question posed by Pratihar and Randrianarisoa in \cite{pratihar2023antipodal}, proving that every antipodal two-weight code is induced by an MRD code over an extension field of $\F_{q}$. 

Moreover, this result allows us to completely characterize $\F_{q^m}$-linear rank-metric codes with two distinct weights. 
\begin{remark} \label{rk:weighttwoweight}
Note that, since a nondegenerate $[n,k]_{q^m/q}$ code always has at least one codeword having rank weight equal to $\min\{m,n\}$ (see \cite[Proposition 3.11]{alfarano2021linear}), a nondegenerate two-weight $[n,k,d]_{q^m/q}$ code admits nonzero weights $d$ and $\min\{m,n\}$. 
\end{remark}

Let us start with the case $k=2$.

\begin{corollary}
Assume that $q \geq m$. Let $\C$ be a nondegenerate two-weight $[n,2,d]_{q^m/q}$ code. Suppose that $\C$ is not antipodal. Then $m-d \mid m$ and $\C$ is a geometric dual of an $\F_{q^{m-d}}$-scattered $[2m-n,2,m-n+d]_{q^m/q}$ code.
\end{corollary}
\begin{proof}
Let $U$ be the system associated with $\C$. Let $d$ and $d'$ be the two nonzero weights of $\C$. By Remark \ref{rk:weighttwoweight}, we know that $d'=m$ and $m<n$, since $\C$ is not antipodal. 
Then by Theorem \ref{th:connection} we obtain that
\[ \dim_{\fq}(U\cap \langle v\rangle_{\F_{q^m}})\in \{n-d,n-m\}, \]
and $n-d,n-m \ne 0$.
The dual $U^{\perp'}$ of $U$ has dimension $2m-n < m$ and 
\[ \dim_{\fq}(U^{\perp'}\cap \langle v\rangle_{\F_{q^m}})\in \{0,m-d\}. \]
We can then use Corollary \ref{cor:linearityoneweight}, obtaining that $U^{\perp'}$ is $\F_{q^{m-d}}$-linear, i.e. $U^{\perp'}$ is a scattered $\F_{q^{m-d}}$-subspace.
\end{proof}

We can now extend the above result to the remaining dimensions.

\begin{theorem}\label{th:classtwoweight}
Assume that $q \geq m$.
Let $\C$ be a nondegenerate two-weight $[n,k,d]_{q^m/q}$ code with $k>2$. Then $m-d \mid m$, $m-d \mid n$ and a system $U$ associated with $\C$ is the dual of an $\F_{q^{m-d}}$-scattered subspace in $\F_{q^m}^k$.
In particular, $n\geq \frac{mk}2$.
Moreover,
\begin{equation*} \label{eq:bound_n}
  \frac{km}{2(m-d)}  \leq \frac{n}{m-d} \leq \frac{mk}{m-d}-1. 
\end{equation*}
\end{theorem}
\begin{proof}
By Remark \ref{rk:weighttwoweight}, we know that the nonzero weights of $\C$ are $d$ and $m=\min\{m,n\}$, and $m<n$, since $\C$ is not antipodal (otherwise $k=2$). 
Let $U\subseteq \mathbb{F}_{q^m}^k$ be a system associated with the code $\C$. 
Then by Theorem \ref{th:connection} we obtain that $\dim_{\fq}(U\cap H)\in \{n-m,n-d\}$, 
for any hyperplane of $\F_{q^m}^k$.
By Proposition \ref{prop:weightdual}, it follows that
\[ \dim_{\fq}(U^{\perp'}\cap H^\perp)\in \{0,m-d\}, \]
or in other words, 
\[ \dim_{\F_q}(U^{\perp'} \cap \langle w \rangle_{\F_{q^m}})\in \{0,m-d\}, \]
for any $w \in \F_{q^m}^k$.  Let $S=\langle U^{\perp'} \rangle_{\F_{q^m}}\subseteq \F_{q^m}^k$ and let $\dim_{\F_{q^m}} (S)=h$. If $\dim_{\F_{q}}(U^{\perp'}) >(h-1)m$ then 
\[\dim_{\fq}(U^{\perp'}\cap \langle w\rangle_{\F_{q^m}})= m-d\] for any $w \in S\setminus \{\underbar 0\}$, and hence
\[ (q^{m-d}-1)\frac{q^{hm}-1}{q^{m}-1}= q^{km-n}-1\]
which gives $d=0$, a contradiction. This means that $\dim_{\F_{q}}(U^{\perp'}) \leq (h-1)m$ and hence by Corollary \ref{cor:linearityoneweight} we get, $m-d \mid m$, $m-d \mid (km-n)$ and $U^{\perp'}$ is an $\F_{q^{m-d}}$-scattered subspace of dimension $\frac{km-n}{m-d}$ of $S\subseteq \F_{q^m}^k$. In particular, applying the bound of Theorem \ref{th:boundscattrank} to $U^{\perp'}$, we get $\frac{km-n}{m-d}\leq \frac{hm}{2(m-d)}$, and hence $n\geq \frac{mk}2$. 
Finally, since $\C$ is a two-weight code, $n<km$ and hence $n\leq km-(m-d)$. 
\end{proof}

In the above theorem, we classified all the two-weight rank-metric codes as those codes arising from the dual of scattered subspaces over extensions of $\fq$. In other terms, such codes arise from a rank-metric code in $\mathbb{F}_{q'^{m'}}^{n'}$, where $q'=q^{m-d}$, $m'=m/(m-d)$ and $n'=n/(m-d)$, see \cite{polverino2023divisible}.




We will now show that, under certain assumptions, two-weight rank-metric codes exists for any length satisfying the above bound.

\begin{theorem}
Let $m,k$ and $d$ be positive integers such that $d<m$, $(m-d)\mid m$ and $mk/(m-d)$ is even, then for any $q$ and for any $n$ multiple of $m-d$ satisfying 
\[
\frac{km}{2(m-d)}  \leq \frac{n}{m-d} \leq \frac{mk}{m-d}-1,
\]
there exists a nondegenerate two-weight $[n,k,d]_{q^m/q}$ code $\C$.
\end{theorem}
\begin{proof}
Because of Theorem \ref{th:existencemaxscatt}, there exists at least one scattered $\F_{q^{m-d}}$-subspace $W$ with dimension $t$ in $\F_{(q^{m-d})^{\frac{m}{m-d}}}^k$ for any 
\[1\leq t \leq  mk/2(m-d).
 \]
 Note that, since $d<m$ and $W$ is a scattered $\F_{q^{m-d}}$-subspace,   $W$ does not contain any $\F_{q^m}$-subspace of $\F_{q^m}^k$. This implies that any code associated with $U=W^{\perp'}$ is a non-degenerate  two-weight $[n,k,d]_{q^m/q}$ code with $n=km-t(m-d)$.
\end{proof}
 
Another natural question regards whether or not the \emph{puncturing} of a two-weight rank-metric code preserves the property of having exactly two weights. As we will see in the next result, this property strongly depends on whether or not the subspace is \emph{maximally scattered}.

\begin{remark}
The puncturing operation has been defined for rank-metric codes of matrices by multiplying every codeword by a fixed full-rank (not necessarily square) matrix. In order to have the analogue definition in vector terms, we need to consider a matrix $A \in \F_q^{n\times n'}$, with $n'\leq n$ of rank $n'$. The \textbf{punctured code of} $\mathcal{C}\subseteq \F_{q^m}^n$ via the matrix $A$ is
\[ \mathcal{C}'=\{ (c_1,\ldots,c_n)A \colon (c_1,\ldots,c_n)\in \C  \}. \]
Note that if $\mathcal{C}$ is $\F_{q^m}$-linear, then $\mathcal{C}'$ is $\F_{q^m}$-linear as well. Moreover, if $\mathcal{C}$ and $\mathcal{C}'$ are have the same dimension, then the system associated with $\mathcal{C}'$ is contained in the system associated with $\mathcal{C}$.
\end{remark}

\begin{theorem}
Assume that $q \geq m$. Let $\C$ be a nondegenerate  $[n,k,d]_{q^m/q}$ code with $d=m-1$.
Let $U$ be an its associated system.
If $U^{\perp'}$ is maximally scattered then a punctured code of $\C$ which has dimension $k$ is either a two-weight code such that an its associated system is the dual of an $\F_{q^e}$-scattered subspace, with $e>1$, or it has more than three weights.
\end{theorem}
\begin{proof}
Let $W$ be the system of a punctured code of $\C$. In particular, this means that $W$ is contained in $U$ and hence $W^{\perp'}\supset U^{\perp'}$.
Since $U^{\perp'}$ is maximally scattered, then $W^{\perp'}$ is not scattered over $\fq$. So, either $W^{\perp'}$ is a $\F_{q^e}$-scattered subspace (for some positive integer $e$) or it is not scattered.
This implies that the code associated with $W$ either it is a two-weight code associated with a scattered subspace over a proper extension of $\fq$ or it has more than three nonzero weights.
\end{proof}

Apart from the maximum scattered subspaces (which give MRD codes), very few examples of maximally scattered $\fq$-subspaces are known, see \cite[Example 3.2]{lavrauw2016scattered}.

\section*{Acknowledgements}

The research was supported by the project COMBINE of ``VALERE: VAnviteLli pEr la RicErca" of the University of Campania ``Luigi Vanvitelli'' and was partially supported by by the INdAM - GNSAGA Project \emph{Tensors over finite fields and their applications}, number E53C23001670001.
This research was also supported by Bando Galileo 2024 – G24-216 and by the project ``The combinatorics of minimal codes and security aspects'', Bando Cassini.

\bibliographystyle{abbrv}
\bibliography{biblio}

\end{document}